\documentclass[prl,superscriptaddress,twocolumn]{revtex4}

\usepackage{amsfonts,amssymb,amsmath}
\usepackage[]{graphics,graphicx,epsfig}
\usepackage{amsthm}

\usepackage{mathtools}
\newcommand{\defeq}{\vcentcolon=}

\def\identity{\leavevmode\hbox{\small1\kern-3.8pt\normalsize1}}

\newtheorem{mydef}{Definition}
\renewcommand{\epsilon}{\varepsilon}

\newtheorem{definition}{Definition} 
\newtheorem{definition2}{Definition}
\newtheorem{prop}[definition2]{Proposition}
\newtheorem{lemma}[definition]{Lemma}

\newtheorem{thm}[definition]{Theorem}

\newtheorem{obs}{Observation}
\newtheorem{claim}[definition2]{Claim}
\newtheorem*{example}{Example}

\newtheorem*{rep@theorem}{\rep@title}
\newcommand{\newreptheorem}[2]{%
\newenvironment{rep#1}[1]{%
 \def\rep@title{#2 \ref{##1} (restatement)}%
 \begin{rep@theorem}}%
 {\end{rep@theorem}}}
\makeatother

\newreptheorem{thm}{Theorem}
\newreptheorem{lem}{Lemma}

\def\ba#1\ea{\begin{align}#1\end{align}}
\def\ban#1\ean{\begin{align*}#1\end{align*}}

\newcommand{\be}{\begin{equation}}
\newcommand{\ee}{\end{equation}}

\def\benum{\begin{enumerate}}
\def\eenum{\end{enumerate}}

\def\squareforqed{\hbox{\rlap{$\sqcap$}$\sqcup$}}
\def\qed{\ifmmode\squareforqed\else{\unskip\nobreak\hfil
\penalty50\hskip1em\null\nobreak\hfil\squareforqed
\parfillskip=0pt\finalhyphendemerits=0\endgraf}\fi}
\def\endenv{\ifmmode\;\else{\unskip\nobreak\hfil
\penalty50\hskip1em\null\nobreak\hfil\;
\parfillskip=0pt\finalhyphendemerits=0\endgraf}\fi}

\newcommand{\<}{\langle}
\renewcommand{\>}{\rangle}

\def\be{\begin{equation}}
\def\ee{\end{equation}}
\def\ben{\begin{eqnarray}}
\def\een{\end{eqnarray}}

\def\bei{\begin{itemize}}
\def\eei{\end{itemize}}

\mathchardef\ordinarycolon\mathcode`\:
\mathcode`\:=\string"8000
\def\vcentcolon{\mathrel{\mathop\ordinarycolon}}
\begingroup \catcode`\:=\active
  \lowercase{\endgroup
  \let :\vcentcolon
  }

\newcommand{\nc}{\newcommand}
 \nc{\proj}[1]{|#1\rangle\!\langle #1 |} 
\nc{\avg}[1]{\langle#1\rangle}

\nc{\conv}{\operatorname{conv}}
\nc{\smfrac}[2]{\mbox{$\frac{#1}{#2}$}} \nc{\Tr}{\operatorname{Tr}}
\nc{\ox}{\otimes} \nc{\dg}{\dagger} \nc{\dn}{\downarrow}
\nc{\lmax}{\lambda_{\text{max}}}
\nc{\lmin}{\lambda_{\text{min}}}

\nc{\csupp}{{\operatorname{csupp}}}
\nc{\qsupp}{{\operatorname{qsupp}}} \nc{\var}{\operatorname{var}}
\nc{\rar}{\rightarrow} \nc{\lrar}{\longrightarrow}
\nc{\poly}{\operatorname{poly}}
\nc{\polylog}{\operatorname{polylog}} \nc{\Lip}{\operatorname{Lip}}
\nc{\Om}{\Omega}
\nc{\wt}[1]{\widetilde{#1}}

\def\>{\rangle}
\def\<{\langle}

\def\a{\alpha}

\nc{\glneq}{{\raisebox{0.6ex}{$>$}  \hspace*{-1.8ex} \raisebox{-0.6ex}{$<$}}}
\nc{\gleq}{{\raisebox{0.6ex}{$\geq$}\hspace*{-1.8ex} \raisebox{-0.6ex}{$\leq$}}}

\nc{\vholder}[1]{\rule{0pt}{#1}}
\nc{\wh}[1]{\widehat{#1}}
\nc{\h}[1]{\widehat{#1}}

\nc{\ob}[1]{#1}

\def\beq{\begin {equation}}
\def\eeq{\end {equation}}

\def\be{\begin{equation}}
\def\ee{\end{equation}}

\nc{\eq}[1]{(\ref{eq:#1})} 
\nc{\eqs}[2]{\eq{#1} and \eq{#2}}

\nc{\eqn}[1]{Eq.~(\ref{eqn:#1})}
\nc{\eqns}[2]{Eqs.~(\ref{eqn:#1}) and (\ref{eqn:#2})}

\nc{\region}{\cS\cW}

\newenvironment{protocol*}[1]
  {
    \begin{center}
      \hrulefill\\
      \textbf{#1}
  }
  {
    \vspace{-1\baselineskip}
    \hrulefill
    \end{center}
  }

\begin{document}

\title{Strong monogamies of no-signaling violations for bipartite correlation Bell inequalities}
\author{Ravishankar \surname{Ramanathan}}
\email{ravishankar.r.10@gmail.com}
\affiliation{National Quantum Information Center of Gdansk,  81-824 Sopot, Poland}
\affiliation{University of Gdansk, 80-952 Gdansk, Poland}
\author{Pawe{\l} \surname{Horodecki}}
\affiliation{National Quantum Information Center of Gdansk, 81-824 Sopot, Poland}
\affiliation{Faculty of Applied Physics and Mathematics, Technical University of Gdansk, 80-233 Gdansk, Poland}

\begin{abstract}
The phenomenon of monogamy of Bell inequality violations is interesting both from the fundamental perspective as well as in cryptographic applications such as the extraction of randomness and secret bits. 
In this article, we derive new and stronger monogamy relations for violations of Bell inequalities in general no-signaling theories. These relations are applicable to the class of binary output correlation inequalities known as XOR games, and to a restricted set of unique games. In many instances of interest, we show that the derived relation provides a significant strengthening over previously known results. The result involves a shift in paradigm towards the importance in monogamies of the number of inputs of one party which lead to a contradiction from local realistic predictions.


\end{abstract}
\maketitle

{\it Introduction.} 
The violation of a Bell inequality is a defining feature of quantum correlations that distinguishes them from classical correlations. Apart from the fundamental interest in the absence of a local hidden variable description of Nature, this feature of quantum theory has led to numerous applications in communication protocols, including key distribution \cite{Ekert, BHK} and generation of randomness \cite{Renner, Pironio}. Many interesting properties of this phenomenon dubbed non-locality have been discovered leading to a better understanding of the set of quantum correlations \cite{Masanes, distillation} which is crucial in the development of further applications. It is known that the set of quantum correlations is a convex set sandwiched between the classical polytope and the no-signaling polytope which is the set of correlations obeying the so-called no-signaling principle (impossibility of faster-than-light communication). 

The violation of Bell inequalities in general no-signaling theories (which includes quantum theory) displays a very interesting property called monogamy. Strong non-local correlations between two parties in extreme cases lead to weak correlations between these parties and any other no-signaling system. Specifically, there are instances where the violation of a Bell inequality by Alice and Bob precludes its violation by Alice and any other party Charlie, when Alice uses the same measurement results in both experiments. This phenomenon is seen to be important in secure communication protocols for key distribution or randomness generation between these parties, due to the fact that any third party such as an eavesdropper is only able to establish weak correlations with their systems \cite{Renner, Kent, Remig}. 

The monogamy of no-signaling correlations was first discovered for the CHSH inequality in \cite{Toner} and has since been shown to be a generic feature of all no-signaling theories \cite{Masanes}. A general monogamy relation applicable to any no-signaling correlations in the two-party Bell scenario that only involves the number of settings (inputs) for one party was developed in \cite{Terhal, Pawlowski}. This approach follows the idea of symmetric extensions and shareability of correlations which can also be recast in terms of the existence of joint probability distributions \cite{Kaszlikowski}. See also \cite{Brandao} for a novel method to deriving monogamies based on the powerful machinery of de Finetti theorems for no-signaling probability distributions. In this paper, we follow a different approach and derive a strengthened version of the monogamy relation for two-party inequalities that holds in many flagship scenarios involving correlation expressions, in particular to the wide class of binary output correlation scenarios known as XOR games and to a restricted set of the so-called unique games. We also investigate the validity of the monogamy relation and the strengthening it provides in many relevant scenarios. The result highlights the importance in monogamies of the number of settings of one party that lead to a contradiction with local realistic predictions.




{\it Monogamy relations for Bell inequality violations from no-signaling constraints.}
The general two-party Bell inequality is written as 
\begin{eqnarray}
&&\mathcal{B}_{AB} \defeq \mathbb{B}_{AB}.\{P(a,b|x,y)\} = \nonumber \\ &&\sum_{x=1}^{m_A} \sum_{y=1}^{m_B}  \sum_{a, b = 1}^{d} \mu(x, y) V(a, b| x, y) P(a, b | x, y) \leq \mathit{R}_L(\mathbb{B}) \nonumber 
\end{eqnarray}
where $m_A (m_B)$ denotes the number of settings for Alice (Bob) and $d = d_A, d_B$ denotes the number of outcomes of both Alice and Bob, with redundant outcomes added if necessary to make $d_A = d_B$. The probability distribution with which the inputs are chosen is denoted by $\mu(x, y)$ ($0 \leq \mu(x,y) \leq 1$), in many situations this distribution is taken to be of product form $\mu(x,y) = \mu^A(x) \mu^B(y)$ reflecting the independence of Alice and Bob in choosing their measurement settings. The $V(a, b| x, y)$ picks out the probabilities $P(a,b|x,y)$ with appropriate coefficients that enter the Bell inequality. The bound $\mathit{R}_L(\mathbb{B})$ is the maximum attainable value of the Bell expression by local deterministic boxes and consequently by their mixtures. The Bell vector $\mathbb{B}_{AB}$ has entries $\mathbb{B}_{AB}(a,b,x,y) = \mu(x,y) V(a,b|x,y)$ and the no-signaling box $\{P(a,b|x,y)\}$ describes the output distributions for different inputs. All the considerations in this article are confined to two-party Bell inequalities, so we will skip the suffix $AB$ wherever it is not required. All the proofs are deferred to the Supplemental Material. 

We consider the scenario where Alice performs the Bell experiment simultaneously with many Bobs $B^{i}$, with $i \in [n]$ for some $n$. A no-signaling constraint across all parties is imposed on the box $P(a, b^1, \dots, b^n | x, y^1, \dots, y^n)$ with inputs $x \in [m_A], y^i \in [m_B]$ and outputs $a, b^i \in [d]$ $\; \; \forall i$, i.e. 
\begin{widetext}
\begin{eqnarray}
\label{no-signal}
\sum_{a=1}^{d} P(a, b^1, \dots, b^n|x, y^1, \dots, y^n) &=& \sum_{a'=1}^{d} P(a', b^1, \dots, b^n|x', y^1, \dots, y^n) \; \; \quad \forall b^i, y^i, x, x' \nonumber \\
\sum_{b^i = 1}^{d} P(a, b^1,\dots, b^i, \dots, b^n | x, y^1, \dots, y^i, \dots, y^n) &=&  \sum_{b'^i=1}^{d} P(a, b^1,\dots, b'^i, \dots b^n | x, y^1, \dots, y'^i, \dots, y^n) \; \; \forall a, x, y^i, y'^i, b^j, y^j (j \neq i). \nonumber \\
\end{eqnarray}
\end{widetext}

Building on the idea of symmetric extensions \cite{Terhal}, it was shown in \cite{Pawlowski} that a monogamy relation holds for arbitrary Bell inequalities in any no-signaling theory in the following form. 
\begin{thm}[TDS03, PB09]
Consider the scenario where Alice wants to perform a certain Bell experiment $\mathbb{B}$ with $K$ Bobs denoted $B^{1}, \dots, B^{K}$ simultaneously. The associated Bell expressions $\mathbb{B}_{AB^{i}}$ for $i \in [K]$ satisfy a monogamy relation when $K = m_B$ i.e. for any box $\{P(a, b^1, \dots, b^{m_B} | x, y^1, \dots, y^{m_B})\}$ with $x \in [m_A], y^i \in [m_B]$ and $a, b^i \in [d]$ $\; \forall i$, in any no-signaling theory it holds that
\begin{eqnarray}
\label{gen-mono}
&&\sum_{i=1}^{m_B} \mathcal{B}_{AB^{i}} = \sum_{i=1}^{m_B} \mathbb{B}_{AB^{i}}.\{P(a,b^1, \dots, b^{m_B} | x, y^1, \dots, y^{m_B})\} \nonumber \\ && \qquad \qquad  \qquad \qquad \qquad \qquad \qquad \qquad \leq m_B \mathit{R}_L(\mathbb{B}) .
\end{eqnarray}
\end{thm}

Notice that the above relation depends upon the inequality only in so far as the number of settings of Bob is concerned, in particular it is independent of $m_A, d$ and $\mu(x,y)$. Quantum theory being a no-signaling theory obeys Eq.(\ref{gen-mono}) and in particular instances is also known to satisfy more stringent monogamy relations \cite{Toner, Kurzynski} due to its additional structure. From here on,  we denote $\mathbb{B}_{AB}.\{P(a,b|x,y)\}$ by $\mathcal{B}$ and $\sum_{i=1}^{N} \mathbb{B}_{AB^{i}}.\{P(a,b^1, \dots, b^{N} | x, y^1, \dots, y^{N})\}$ as $\sum_{i=1}^{N} \mathcal{B}_{AB^{i}}$ for ease of notation. 

In this paper, we derive a no-signaling monogamy relation that is applicable to a wide class of bipartite inequalities involving correlation expressions. To this end, we consider a parameter we call the contradiction number characterizing the Bell expression. This quantity denotes the difference between the number of measurement settings of one party (in this case Bob) and the maximum number of their measurement settings up to which the optimal no-signaling value can be attained by a local deterministic box.

\begin{mydef}[Contradiction number]
\label{def-con}
For any Bell inequality $\mathcal{B}_{AB} \leq \mathit{R}_L \leq \mathit{R}_{NS}$,
denote by $S(\mathbb{B})$ the set of settings of party $B$ of minimum cardinality $\mathit{C}_{\mathbb{B}} \defeq |S({\mathbb{B}})|$ whose removal leads to the optimum no-signaling value being achieved by a local deterministic box i.e. $\mathcal{\underline{B}}^{nc}_{AB} \leq \mathit{\underline{R}}^{nc}_L(\mathbb{\underline{B}}^{nc}) = \mathit{\underline{R}}^{nc}_{NS}(\mathbb{\underline{B}}^{nc})$ where
\begin{eqnarray}
\label{sat-exp}
\mathcal{\underline{B}}^{nc}_{AB} \defeq \sum_{x=1}^{m_A} \sum_{y=1}^{m_B - \mathit{C}_{\mathbb{B}}} \mu(x, y) \sum_{a, b = 1}^{d} V(a, b| x, y) P(a, b | x, y)
\end{eqnarray}
and $\mathit{\underline{R}}^{nc}_L(\mathbb{\underline{B}}^{nc})$ is the optimum local value of the expression $\mathcal{\underline{B}}^{nc}_{AB}$ while $\mathit{\underline{R}}^{nc}_{NS}(\mathbb{\underline{B}}^{nc})$ is its optimum no-signaling value. We then call $\mathit{C}_{\mathbb{B}}$ as the \textsl{contradiction number} for the Bell inequality.
\end{mydef}

The monogamy relation that we derive is given in terms of the contradiction number of the inequality $\mathit{C}_{\mathbb{B}}$ as
\begin{eqnarray}
\label{mono-NS2}
\sum_{i = 1}^{\mathit{C}_{\mathbb{B}} + 1} \mathcal{B}_{AB^{i}} = \sum_{i = 1}^{\mathit{C}_{\mathbb{B}} + 1} &&\mathbb{B}_{AB^{i}}.\{P(a, b^1, \dots, b^{\mathit{C}_{\mathbb{B}}+1} | x, y^1, \dots, y^{\mathit{C}_{\mathbb{B}}+1})\} \nonumber \\ &&\leq (\mathit{C}_{\mathbb{B}} + 1) \mathit{R}_L(\mathbb{B}) .
\end{eqnarray}
Clearly, $C_{\mathbb{B}} \leq m_B - 1$ since a classical winning strategy always exists when considering only a single setting of Bob. As we shall see, in many instances of interest the inequality is strict so that the relation in Eq.(\ref{mono-NS2}) provides a significant strengthening to that in Eq.(\ref{gen-mono}). Note that the above definition could be modified to take into account either Alice or Bob in which case the monogamy relation will hold with the common party being the one that does not achieve the minimum cardinality.

{\it Monogamy of correlation Bell expressions.}
A natural generalization of the CHSH inequality to more inputs $m_A, m_B$ is 
to Bell inequalities involving the correlators $\mathbb{E}_{x,y}$ defined as \cite{Bancal}
\begin{equation}
\mathbb{E}_{x,y} = \sum_{k=0}^{d-1} \lambda_k P(a - b = k \; \texttt{mod} \; \texttt{d} | x, y), 
\end{equation}
with real parameters $\lambda_k$. These "correlation Bell expressions'' do not depend on the individual values assigned to Alice and Bob's outcomes, but rather on how these outcomes relate to each other. These are the most common Bell inequalities and are written in general as
\begin{equation}
\label{corr}
\mathcal{B}^{\doteqdot} \defeq \sum_{x=1}^{m_A} \sum_{y=1}^{m_B} \alpha_{x,y} \mathbb{E}_{x,y} \leq \mathit{R}_L(\mathbb{B}^{\doteqdot}),
\end{equation}
where $\alpha_{x,y}$ are arbitrary reals.

In the scenario of binary outputs ($d = 2$), the correlation inequalities are also known as XOR games. XOR games (for two parties) consider the scenario when the predicate $V(a,b|x,y) \in \{0, 1\}$ only depends upon the xor of the two parties' outcomes $V(a \oplus b|x,y)$. The fact that these are equivalent to correlation inequalities for $d=2$ is simply seen by noting that for $a,b,k \in \{0,1\}$, we have $P(a \oplus b = k |x, y) = \frac{1}{2} \left(1 + (-1)^k \mathbb{E}_{xy} \right)$ where the correlators $\mathbb{E}_{xy}$ are defined using $\lambda_k = (-1)^k$. XOR games are a widely studied and the most well-understood class of Bell inequalities due to the fact that the maximum quantum values of these inequalities are directly calculable by a semi-definite program \cite{Tsirelson, Cleve, Wehner}. Our first main result is the following.


\begin{prop}
Consider the general XOR Bell expression $\mathcal{B}^{\oplus}$
\begin{eqnarray}
&&\mathcal{B}_{AB}^{\oplus} \defeq \mathbb{B}_{AB}^{\oplus}.\{P(a,b|x,y)\} = \nonumber \\
&&\sum_{x=1}^{m_A} \sum_{y=1}^{m_B} \mu(x, y) \sum_{a, b = 0, 1} V(a \oplus b| x, y) P(a, b | x, y) \leq \mathit{R}_L(\mathbb{B}^{\oplus}) \nonumber 
\end{eqnarray}
with corresponding number of contradictions $\mathit{C}_{\mathbb{B}^{\oplus}}$. For any $m_A, m_B, \mu(x,y)$ and any no-signaling box $\{P(a, b^1, \dots, b^{\mathit{C}_{\mathbb{B}^{\doteqdot}}+1} | x, y^1, \dots, y^{\mathit{C}_{\mathbb{B}^{\doteqdot}} + 1}) \}$ with $x \in [m_A], y^i \in [m_B]$ and $a, b^i \in \{0,1\} \; \; \forall i$, we have that $\mathbb{B}^{\oplus}$ satisfies Eq. (\ref{mono-NS2}), i.e., 
%
%
\begin{equation}
\sum_{i=1}^{\mathit{C}_{\mathbb{B}^{\oplus}} +1} \mathcal{B}^{\oplus}_{AB^{i}} \leq (\mathit{C}_{\mathbb{B}^{\oplus}}+1) \mathit{R}_L(\mathbb{B}^{\oplus}). 
\end{equation}

\end{prop}

The interest in the above statement is due to the fact that two-party XOR games still provide the Bell inequality of choice in many non-locality applications. The paradigmatic example of such an XOR game for which Eq.(\ref{mono-NS2}) significantly outperforms previously known monogamy relations is the chained Bell inequality of Braunstein and Caves \cite{BC}. These inequalities have been applied in randomness amplification and key distribution protocols secure against no-signaling adversaries \cite{Renner, Kent, Remig}. The chained Bell expression $\mathcal{B}_{AB}^{ch, N}$ involves two parties with $m_A = m_B = N$ inputs and $d = 2$ outcomes, it is written as 
\begin{eqnarray} 
\sum_{a,b = 0,1} \Big[ &&\sum_{\substack{x, y =1 \\ x = y \vee x = y +1}}^{N} \mu(x,y) P(a \oplus b = 0 |x, y) +  \nonumber \\ && \mu(1,N) P(a \oplus b = 1 |1, N) \Big]
\end{eqnarray}
where $\mathit{R}_L({\mathbb{B}^{ch, N}}) = 1 - \min_{x,y}{\mu(x,y)}$ is the maximum achievable value by a local box.

\begin{example}
For any no-signaling box $\{P(a,b,c|x,y,z)\}$ with $a, b, c \in [d]$ and $x, y, z \in [N]$, the chain Bell inequality satisfies the strong monogamy relation
\begin{equation}
\mathcal{B}_{AB}^{ch,N} + \mathcal{B}_{AC}^{ch, N} \leq 2\mathit{R}_L({\mathbb{B}^{ch, N}}) 
 \end{equation}
\end{example}
The fact that the contradiction number for chain inequalities is one is due to the following observation. 
If Bob fails to measure even one single setting out of the $N$, perhaps due to an imperfect random number generator choosing his inputs \cite{Renner}, all the constraints in the remaining part of the expression can be satisfied by a deterministic box. Explicitly, if Bob fails to measure setting $k \in [N]$ in an N-chain, a classical winning strategy is for Alice to output $a = 0$ for $x \in [k]$ and $a = 1$ for $k < x \leq N$ while Bob outputs $b=0$ for $y \in [k-1]$ and $b=1$ for $k < y \leq N$. In this case therefore, $C_{\mathbb{B}^{ch, N}} +1 = 2$ is much smaller than $m_B = N$, so that Eq.(\ref{mono-NS2}) provides a significant strengthening over Eq.(\ref{gen-mono}).
Interestingly, the monogamy relation also holds when one considers the same family of XOR games, i.e. where the predicate is defined by the same condition but with the inputs being specified over different domain sizes $N$ and $M$, an observation we pursue elsewhere.

The above monogamy relation is the strongest possible in that it implies that any violation of the chain Bell expression by Alice and Bob precludes its violation by Alice and Charlie. In a cryptographic application \cite{Renner, Kent}, this suggests that if Alice and Bob are able to test for strong correlations leading to a violation of the chain inequality, then the correlations that their systems share with an eavesdropper are necessarily weaker. A similar statement also holds for the higher-dimensional generalization of the chain inequalities considered in \cite{Barrett2} which also analogously have contradiction number one. Interestingly though, the very fact that the contradiction number is small provides an attack in this task to the eavesdropper who may for instance tamper with the random number generator in the Bell test so that Bob does not measure a single setting and a local model exists. 

%

{\it Unique games.}
A generalization of XOR games to more outputs is a family of Bell inequalities called unique games \cite{Kempe}, which also involve only certain correlations between Alice and Bob's outputs with no constraint on the marginals. A two-party game is called $unique$ if for every input pair $x, y$ with $x \in [m_A]$ and $y \in [m_B]$ and for every output $a \in [d]$ of Alice, Bob is required to produce a unique output $b \in [d]$ defined by some permutation of Alice's output which can be different for different input pairs, i.e. $V(a, b | x, y) = 1$ if and only if $b = \sigma_{x,y}(a)$ and $0$ otherwise. 
Unique games are an important class of Bell expressions, finding application in the field of computational complexity in the approximation of computationally hard problems. It is easy to see that the maximum value of every unique game with fixed $\{\sigma_{x,y}\}$ can be achieved by a no-signaling box $\{P(a,b|x,y)\}$ with non-zero entries $P(a, \sigma_{x,y}(a) | x,y) = \frac{1}{d} \; \; \forall a \in [d]$. We now investigate the monogamy relations for this important category of Bell expressions. 
In order to do this, we introduce a strengthened version of the concept of contradiction number.
\begin{mydef}[Strong contradiction number]
\label{def-con2}
For any Bell inequality $\mathcal{B}_{AB} \leq \mathit{R}_L \leq \mathit{R}_{NS}$,
denote by $S^{(s)}(\mathbb{B})$ the set of settings of party $B$ of minimum cardinality $\mathit{C}^{(s)}_{\mathbb{B}} \defeq \vert S^{(s)}({\mathbb{B}}) \vert$ whose removal leads to the optimum no-signaling value being achieved by local deterministic boxes i.e. $\mathcal{\underline{B}}^{nc}_{AB} \leq \mathit{\underline{R}}^{nc}_L(\mathbb{\underline{B}}^{nc}) = \mathit{\underline{R}}^{nc}_{NS}(\mathbb{\underline{B}}^{nc})$ where 
\begin{eqnarray}
\label{sat-exp3}
&&\mathcal{\underline{B}}^{nc}_{AB}\defeq \sum_{x=1}^{m_A} \sum_{y=1}^{m_B - \mathit{C}^{(s)}_{\mathbb{B}}} \mu(x, y) \sum_{a, b = 1}^{d} V(a, b| x, y) P(a, b | x, y), \nonumber \\
\end{eqnarray}
and moreover there exists for every $\{a, x\}$ a local deterministic box achieving value $\mathit{\underline{R}}^{nc}_L(\mathbb{\underline{B}}^{nc})$ that deterministically outputs $a$ for input $x$. We then call $\mathit{C}^{(s)}_{\mathbb{B}}$ as the \textsl{strong contradiction number} for the Bell inequality.
\end{mydef}

%

%
%

In the important case of the so-called free unique games, i.e. where $\mu(x,y) = \mu^A(x) \mu^B(y)$ or in the slightly more general scenario when $\mu(x,y) = \beta_{y,y'} \mu(x,y')$ for some parameters $\beta_{y,y'}$ independent of $x$, the strong monogamy relation does hold. For more general unique games, the proofs presented here do not apply and it is unclear whether the strengthened monogamy relation holds, however it does for the specific case when $m_A = 2$, and arbitrary $m_B, d, \mu(x,y)$. 

\begin{prop}
\label{unique-prop}
Any unique game defining an inequality 
\begin{eqnarray}
&&\mathcal{B}_{AB}^{U} \defeq \nonumber \\ &&\sum_{x=1}^{m_A} \sum_{y=1}^{m_B} \sum_{a, b = 1}^{d} \mu(x,y) V(a, \sigma_{x,y}(a)| x, y) P(a, b | x, y) \leq \mathit{R}_L(\mathbb{B}^{U}) \nonumber
\end{eqnarray}
with the restriction that $\mu(x,y) = \beta_{y, y'} \mu(x, y') \; \; \forall x, y, y'$ and with associated strong contradiction number $C^{(s)}_{\mathbb{B}^{U}}$ satisfies Eq.(\ref{mono-NS2}) in any no-signaling theory independently of $m_A, m_B, d$ for arbitrary non-negative parameters $\beta_{y, y'}$ that do not depend on $x$, i.e.
\begin{equation}
\sum_{i=1}^{\mathit{C}^{(s)}_{\mathbb{B}^{U}}+1} \mathcal{B}_{AB^{i}}^{U} \leq (\mathit{C}^{(s)}_{\mathbb{B}^{U}} + 1) \mathit{R}_L(\mathbb{B}^{U}).
\end{equation}
Moreover, every unique game with $m_A = 2$ obeys Eq.(\ref{mono-NS2}) independent of $m_B, d$ and $\mu(x,y)$.
%
%
\end{prop}

We also observe that every non-trivial unique game (with $\mathit{R}_{NS}(\mathbb{B}^U) > \mathit{R}_L(\mathbb{B}^U)$) is monogamous at least to a slight degree in the sense that it is not possible for Alice-Bob and Alice-Charlie to simultaneously achieve the maximum no-signaling value of the game. 
\begin{obs}
For any unique game $\mathbb{B}^U$ with parameters $m_A, m_B, d$ and $\mathit{R}_{NS}(\mathbb{B}^U) > \mathit{R}_L(\mathbb{B}^U)$ and any no-signaling box $\{P(a,b,c|x,y,z)\}$ with $a,b,c \in [d]$ and $x \in [m_A], y,z \in [m_B]$, we have
\begin{equation}
(\mathbb{B}_{AB}^U + \mathbb{B}_{AC}^U).\{P(a,b,c|x,y,z)\} < 2 \mathit{R}_{NS}(\mathbb{B}^U).
\end{equation}
\end{obs}

%
%


%

{\it Discussion.}
A natural question to investigate is whether the strengthened monogamy relations derived here could possibly hold for all inequalities beyond the correlation expressions considered here. Another interesting question is whether all Bell inequality violations are monogamous at least to a slight degree, i.e. is it always the case that when two parties Alice and Bob violate a Bell inequality to its optimal no-signaling value, Alice and Charlie are unable to also optimally violate the inequality? Interestingly, the answer to both questions turns out to be negative due to the following. 

\begin{claim}
\label{counter}
There exists a two-party Bell expression $\mathbb{B}$ with $m_A = m_B = 3, d = 4$, associated contradiction numbers $C_{\mathbb{B}} = C^{(s)}_{\mathbb{B}} = 1$, maximum no-signaling value greater than the maximum local value ($\mathit{R}_{NS} > \mathit{R}_{L}$) and a three-party no-signaling box $\{P(a,b,c|x,y,z)\}$ with $a, b, c \in \{1, 2, 3, 4\}$ and $x, y, z \in \{1, 2, 3\}$ such that 
\begin{equation}
(\mathbb{B}_{AB} + \mathbb{B}_{AC}).\{P(a,b,c|x,y,z)\} = 2 \mathit{R}_{NS}(\mathbb{B}).
\end{equation}
\end{claim}
We present (in the Supplemental Material) the Bell inequality which does not fall within the classes considered in this article and which counteracts the hope that the strengthened monogamy relation for the contradiction number can be applied to arbitrary Bell inequalities beyond those involving correlation expressions. The domain of validity and the tightness of the monogamy relations deserves further investigation.


{\it Conclusions.} In this paper, we have presented a strong monogamy relation that applies to a wide class of Bell inequalities such as XOR games and free unique games. These monogamy relations provide a significant improvement over previously known results and suggest why well-known inequalities such as the Braunstein-Caves chain inequalities are useful in many cryptographic tasks. The methods presented here can also be applied to derive monogamy relations between different Bell inequalities within the same class, an investigation we defer for the future. 
Apart from the quantitative results presented here, the shift in the paradigm from the importance of the number of settings to the number of contradictions leads also to questions regarding the exact physical reasons underpinning these monogamies and their multipartite generalization. One might also consider monogamies in scenarios with multiple Alices and Bobs, in which case in addition to the number of contradictions, the placement and relations between different contradictory settings should also be relevant.

{\em Acknowledgements.} This work is supported by the ERC grant QOLAPS and also forms part of the Foundation for Polish Science TEAM project co-financed by the EU European Regional Development Fund. We thank M. L. Nowakowski, M. Horodecki and K. Horodecki for useful discussions.

{\bf Supplemental Material.}
Here we present the proofs of the propositions in the main text.

\textbf{Proposition 1.}
\textit{The general two-party XOR inequality 
\begin{eqnarray}
&&\mathbb{B}_{AB}^{\oplus}.\{P(a,b|x,y)\} \defeq \nonumber \\
&&\sum_{x=1}^{m_A} \sum_{y=1}^{m_B} \mu(x, y) \sum_{a, b = 0, 1} V(a \oplus b| x, y) P(a, b | x, y) \leq \mathit{R}_L(\mathbb{B}^{\oplus}) \nonumber 
\end{eqnarray}
satisfies in any no-signaling theory the monogamy relation Eq. (\ref{mono-NS2}) for arbitrary $m_A, m_B$ and $\mu(x, y)$, i.e.
\begin{equation}
\label{xor-mono-2}
\sum_{i=1}^{\mathit{C}_{\mathbb{B}^{\oplus}} +1} \mathcal{B}^{\oplus}_{AB^{i}} \leq (\mathit{C}_{\mathbb{B}^{\oplus}}+1) \mathit{R}_L(\mathbb{B}^{\oplus}). 
\end{equation}}
\begin{proof}
Let $\{Q(a, b^1, \dots, b^{\mathit{C}+1} | x, y^1, \dots, y^{\mathit{C}+1})\}$ be a no-signaling box that achieves the maximum value of the expression 
\begin{equation}
\label{corr-sum}
\sum_{i=1}^{\mathit{C}+1} \mathcal{B}^{\oplus}_{AB^{i}} = \sum_{i=1}^{\mathit{C} +1} \mathbb{B}^{\oplus}_{AB^{i}}.\{P(a, b^1, \dots, b^{\mathit{C}+1} | x, y^1, \dots, y^{\mathit{C} + 1}) \}.
\end{equation}
Note that we have dropped the subscript ${\mathbb{B}^{\oplus}}$ on $\mathit{C}$ for ease of notation. 

Now the sum $\sum_i \mathcal{B}^{\oplus}_{AB^{i}}$ can be rewritten as $\sum_i \mathcal{B}^{\oplus}_{AB^{i}} = \sum_i \mathcal{\underline{B}}^{\oplus}_i$ with $\mathcal{\underline{B}}^{\oplus}_i \defeq \mathcal{\underline{B}}^{nc, \oplus}_{i} + \mathcal{\underline{B}}^{c, \oplus}_{i}$. Here 
\begin{eqnarray}
\label{sat-exp2}
&&\mathcal{\underline{B}}^{nc, \oplus}_{i} = \mathbb{\underline{B}}^{nc, \oplus}_{i}.\{P(a,b^i|x,y^i)\} \defeq \nonumber \\ &&\sum_{x=1}^{m_A} \sum_{y^i=1}^{m_B - \mathit{C}} \mu(x, y^i) \sum_{a, b^i = 0, 1} V(a \oplus b^i| x, y^i) P(a, b^i | x, y^i) \nonumber \\
\end{eqnarray} 
is an expression analogous to Eq.(\ref{sat-exp}) in Definition \ref{def-con} whose optimal no-signaling value is achieved by a local deterministic box, and $\mathcal{\underline{B}}^{c, \oplus}_{i}$ is defined by
\begin{eqnarray}
\label{sing-set}
&&\mathcal{\underline{B}}^{c, \oplus}_{i} = \mathbb{\underline{B}}^{c, \oplus}_{i}.\{P(a,b^1, \dots, b^{\mathit{C}+1} | x, y^1, \dots, y^{\mathit{C}+1})\} \defeq \nonumber \\ && \sum_{j=i+1}^{i+\mathit{C}} \sum_{x=1}^{ m_A}  \mu(x, y^j) \sum_{a, b^j = 0, 1} V(a \oplus b^j| x, y^j) P(a, b^j | x, y^j), \nonumber \\
\end{eqnarray}
where the sum over $j$ is over the parties labeled $i+1$ to $i + C$ modulo the number of parties. Consequently in the expression in Eq. (\ref{sing-set}), the $C$ contradictions are spread over the $C$ remaining parties labeled $j$ ($\neq i$) and each of these parties only measures a single setting.

Firstly, let us note that the maximum achievable value by local deterministic boxes for the expression $\mathbb{B}^{\oplus}_{AB^{i}}$ is equal to that for the expression $\mathbb{\underline{B}}^{\oplus}_i$, i.e.
\begin{equation}
\label{local-eq}
\mathit{R}_L(\mathbb{B}^{\oplus}_{AB^{i}} ) = \mathit{\underline{R}}_L(\mathbb{\underline{B}}^{\oplus}_{i}).
\end{equation}
This is because any deterministic box that achieves the local maximum for $\mathbb{B}^{\oplus}_{AB^{i}}$ assigns some specific outcome $b^{j}_{*}$ for each of Bob's inputs in the set of contradictions $S(\mathbb{B}^{\oplus}_{AB})$ which can equivalently be ascribed when these same inputs are measured by different parties $j \neq i$. Similarly, any deterministic strategy for the Bell scenario in $\mathbb{\underline{B}}^{\oplus}_{i}$ can be rewritten as a deterministic strategy for that in $\mathbb{B}^{\oplus}_{AB^{i}}$. 

Now, consider any no-signaling box $\{\underline{Q}(a,b^1, \dots, b^{C+1} | x, y^1, \dots, y^{C+1})\}$ that achieves the optimal no-signaling value for $\mathcal{\underline{B}}^{\oplus}_i$. This box may without loss of generality be written as
\begin{eqnarray}
\label{opt-form}
&&\{\underline{Q}(a,b^1, \dots, b^{C+1} | x, y^1, \dots, y^{C+1})\} = \nonumber \\ &&\otimes_{\substack{j=1 \\ j\neq i}}^{C+1} \{\underline{Q}^j(b^j|y^j)\} \otimes \{\underline{Q}^{i}(a,b^i|x, y^i)\}.
\end{eqnarray}
where each of the boxes $\{\underline{Q}^j(b^j|y^j)\}$ with $j \neq i$ deterministically return a single output $b^{j}_{*}$ for input $y^j$ and $\{\underline{Q}^{i}(a,b^i|x,y^i)\}$ is any no-signaling box that could in principle be non-local. The reason for this structure is that the optimization of the linear Bell function over the no-signaling convex polytope attains the optimum at some extreme point so we do not need to consider probabilistic mixtures of the boxes $\{\underline{Q}^j(b^j|y^j)\}$. Our task is to show that $\{\underline{Q}^{i}(a,b^i|x,y^i)\}$ can be replaced by a local box so that the whole box $\underline{Q}$ takes on local structure and the optimal no-signaling value of $\mathcal{\underline{B}}^{\oplus}_i$ becomes equal to its optimal local value. 

To this end, observe that the value $\mathcal{\underline{B}}^{nc, \oplus}_{i}$ only depends on $\{\underline{Q}^{i}(a,b^i|x,y^i)\}$ while the value $\mathcal{\underline{B}}^{c, \oplus}_{i}$ depends on the rest of the boxes labeled $j$ as well as the marginal distributions $\{\underline{Q}^A(a|x)\}$ with entries $\underline{Q}^A(a|x) = \sum_{b^i} \underline{Q}^{i}(a,b^i|x,y^i)$. We now prove the following useful statement.

\begin{lemma}
\label{local-marg-xor}
The maximum value of $\mathcal{\underline{B}}^{nc, \oplus}$ over all no-signaling boxes under the constraint of an arbitrary fixed marginal distribution $\{\tilde{P}^{A}(a|x)\}$ is achieved by a local box $\sum_k p_k \{P^{A}_{k}(a|x)\} \otimes \{P^{B}_{k}(b|y)\}$ with $$\sum_k p_k P^{A}_k(a|x) = \tilde{P}^{A}(a|x) \; \; \; \forall a, x.$$
\end{lemma}
\begin{proof}
By definition without a restriction on the marginals, $\mathcal{\underline{B}}^{nc, \oplus}$ is maximized by a local deterministic box $\{P_{D}^{1}(a,b|x,y)\}$. A local relabeling of the outputs by Alice and Bob using shared randomness of the form $a \rightarrow a \oplus 1, b \rightarrow b \oplus 1$ for all $x, y$ leads to a local box $\{P_{D}^{2}(a,b|x,y)\}$ which also achieves the maximum no-signaling value (since this value only depends on the xor of the outcomes which is unchanged under this operation), i.e.
\begin{eqnarray}
P_{D}^{2}(a \oplus 1, b \oplus 1| x, y) \defeq P_{D}^{1}(a, b | x, y) \; \; \forall a, b, x, y. \nonumber
\end{eqnarray}
This implies that the local box  
\begin{equation}
\{P^{\texttt{U}}_{L}(a, b | x, y)\} \defeq  \frac{1}{2} \sum_{k=1,2} \{P_{D}^{k}(a, b | x, y)\}
\end{equation}
with uniform marginals 
\begin{eqnarray}
P^{\texttt{U}, A}_{L}(a | x) &=& \sum_{b} P^{\texttt{U}}_{L}(a, b | x, y) \nonumber \\
&=& \frac{1}{2} \; \; \; \forall a, x
\end{eqnarray}
also achieves the maximum no-signaling value of $\mathcal{\underline{B}}^{nc, \oplus}$ when there is no constraint of a fixed marginal distribution.

Now, let $\{P^*(a,b|x,y)\}$ denote the box that achieves the maximum no-signaling value of $\mathcal{\underline{B}}^{nc, \oplus}$ under the restriction of a fixed marginal distribution 
\begin{equation}
\sum_{b} P^*(a,b|x,y) = \tilde{P}^A(a|x) \;\; \; \forall a, x.
\end{equation}
Let us express $\{P^*(a,b|x,y)\}$ as a convex combination of the vertices of the no-signaling polytope (including the local and non-local vertices).
\begin{eqnarray}
\label{conv-decomp}
&&\{P^*(a,b|x,y)\} = \sum_v p^*_v \{P_v(a,b|x,y)\} \nonumber \\
&&= \sum_{v_L} p^*_{v_L} \{P_{v_L}(a,b|x,y)\} + \sum_{v_{NL}} p^*_{v_{NL}} \{P_{v_{NL}}(a,b|x,y)\}. \nonumber \\
\end{eqnarray}
where the sum over $v_L$ is over the local vertices and the sum over $v_{NL}$ is over the non-local ones.
The vertices of the no-signaling polytope for the situation $d=2$ and arbitrary $m_A, m_B$ have been characterized in \cite{Jones}. It was shown that every vertex of this polytope has entries $P(a|x), P(b|y) \in \{0, \frac{1}{2}, 1\} \; \; \forall a, b, x, y$ with the non-local vertices having entries $0, \frac{1}{2}, 1$ while the local vertices only have the deterministic entries $0, 1$. Due to this property and following the construction above we now see that we can replace every non-local vertex $\{P_{v_{NL}}(a,b|x,y)\}$ in the decomposition by a local box $\{P^{\texttt{U}, v_{NL}}_{L}(a, b | x, y)\}$. This is done by replacing the part of the non-local vertex with entries $1/2$ by a uniform combination of the same part of the deterministic boxes $\{P_{D}^{1}(a,b|x,y)\}$ and $\{P_{D}^{2}(a,b|x,y)\}$. Moreover since these boxes by definition maximize the Bell expression, such an operation can only increase the Bell value. In other words, this construction leads to the local box
\begin{eqnarray}
\label{conv-decomp}
&&\{P^{*}_{L}(a,b|x,y)\} \defeq \nonumber \\ &&\sum_{v_L} p^*_{v_L} \{P_{v_L}(a,b|x,y)\} + \sum_{v_{NL}} p^*_{v_{NL}} \{P^{\texttt{U}, v_{NL}}_{L}(a, b | x, y)\} \nonumber \\
\end{eqnarray}
which achieves the same marginal distribution 
\begin{equation}
\sum_{b} P^{*}_{L}(a,b|x,y) = \tilde{P}^A(a|x) \;\; \; \forall a, x
\end{equation}
and 
\begin{equation}
\mathbb{\underline{B}}^{nc, \oplus}.\{P^{*}_{L}(a,b|x,y)\} \geq \mathbb{\underline{B}}^{nc, \oplus}.\{P^{*}(a,b|x,y)\}.
\end{equation}
This proves the statement in Lemma \ref{local-marg-xor}.
\end{proof}
Returning to the box $\{\underline{Q}(a,b^1, \dots, b^{C+1} | x, y^1, \dots, y^{C+1})\}$ that achieves the optimal no-signaling value for $\mathcal{\underline{B}}^{\oplus}_i$ and its general form in Eq.(\ref{opt-form}), we now see that one can use Lemma\ref{local-marg-xor} to construct a local box by replacing $\{\underline{Q}^i(a,b^i|x, y^i)\}$ by a local box $\{\underline{Q}^i_{L}(a,b^i|x, y^i)\}$ with the same marginal distribution which achieves a not smaller value of $\mathcal{\underline{B}}^{nc, \oplus}$, i.e.
\begin{eqnarray}
\label{opt-form-2}
&&\{\underline{Q}_L(a,b^1, \dots, b^{C+1} | x, y^1, \dots, y^{C+1})\} = \nonumber \\ &&\otimes_{\substack{j=1 \\ j\neq i}}^{C+1} \{\underline{Q}^j(b^j|y^j)\} \otimes \{\underline{Q}^{i}_L(a,b^i|x, y^i)\}.
\end{eqnarray}
Since $\mathcal{\underline{B}}^{c, \oplus}_{i}$ only depends on $\{\underline{Q}^i(a,b^i|x, y^i)\}$ via the marginal $\{\underline{Q}^A(a|x)\}$, we see that the maximum no-signaling value of $\mathcal{\underline{B}}^{\oplus}_{i}$ is achieved by the local box $\{\underline{Q}_L(a,b^1, \dots, b^{C+1} | x, y^1, \dots, y^{C+1})\}$. This value is by definition $\mathit{\underline{R}}_L(\mathbb{\underline{B}}^{\oplus})$ which by Eq.(\ref{local-eq}) is the same as $\mathit{R}_L(\mathbb{B}^{\oplus}_{AB} )$. Therefore, 
\begin{equation}
\sum_{i=1}^{\mathit{C}+1} \mathcal{\underline{B}}^{\oplus}_i \leq (\mathit{C}+1) \mathit{\underline{R}}_L(\mathbb{\underline{B}}^{\oplus})
\end{equation}
which is equivalent to Eq.(\ref{xor-mono-2}).
\end{proof}

\textbf{Proposition 2.}
\textit{Any unique game defining an inequality 
\begin{eqnarray}
\label{unique-game}
&&\mathcal{B}_{AB}^{U} \defeq \nonumber \\ &&\sum_{x=1}^{m_A} \sum_{y=1}^{m_B} \mu(x,y) \sum_{a, b = 1}^{d} V(a, \sigma_{x,y}(a)| x, y) P(a, b | x, y) \leq \mathit{R}_L \nonumber \\
\end{eqnarray}
with the restriction that $\mu(x,y) = \beta_{y, y'} \mu(x, y') \; \; \forall x, y, y'$ and with associated strong contradiction number $C^{(s)}_{\mathbb{B}^{U}}$ satisfies Eq.(\ref{mono-NS2}) in any no-signaling theory independently of $m_A, m_B, d$ for arbitrary non-negative parameters $\beta_{y, y'}$ that do not depend on $x$, i.e.
\begin{equation}
\sum_{i=1}^{\mathit{C}^{(s)}_{\mathbb{B}^{U}}+1} \mathcal{B}_{AB^{i}}^{U} \leq (\mathit{C}^{(s)}_{\mathbb{B}^{U}} + 1) \mathit{R}_L(\mathbb{B}^{U}).
\end{equation}
Moreover, every unique game with $m_A = 2$ obeys Eq.(\ref{mono-NS2}) independent of $m_B, d$ and $\mu(x,y)$.}
%
%

\begin{proof}
We first prove the part of the Proposition concerning the situation $\mu(x,y) = \beta_{y,y'} \mu(x,y') \; \; \forall x, y, y'$ and arbitrary $m_A, m_B, d$.
Consider the unique game Eq.(\ref{unique-game}) played between Alice and $\mathit{C}^{(s)}+1$ Bobs (where we have dropped the subscript $\mathbb{B}^U$ for simplicity) and let $\{Q(a, b^1, \dots, b^{\mathit{C}^{(s)}+1} | x, y^1, \dots, y^{\mathit{C}^{(s)}+1})\}$ be the box that maximizes the resulting expression $\sum_{i=1}^{C^{(s)}+1} \mathcal{B}^{U}_{AB^{i}}$
over all no-signaling boxes.  

Now, as in the previous proof the sum $\sum_i \mathcal{B}^{U}_{AB^{i}}$ can be rewritten as $\sum_i \mathcal{\underline{B}}^{U}_i$ with $\mathcal{\underline{B}}^{U}_i = \mathcal{\underline{B}}^{U, nc}_i + \mathcal{\underline{B}}^{U, c}_i$. Here
\begin{eqnarray}
\label{unique-no-con}
&&\mathbb{\underline{B}}^{U, nc}.\{P(a,b|x,y)\} \defeq \nonumber \\ &&\sum_{x=1}^{m_A} \sum_{y=1}^{m_B - C^{(s)}} \mu(x,y) \sum_{a, b = 1}^{d} V(a, \sigma_{x,y}(a)| x, y) P(a, b | x, y) \nonumber 
\end{eqnarray}
is the expression analogous to Eq. (\ref{sat-exp3}) in Definition \ref{def-con2} with no contradictions saturated by a local deterministic box (for every outcome $a$ for every setting $x$ of Alice). And 
\begin{eqnarray}
\label{unique-con}
&& \mathbb{\underline{B}}^{U, c}_i.\{P(a,b^1, \dots, b^{C+1}|x,y^1, \dots, y^{C+1})\} \defeq \nonumber \\ && \sum_{j=i+1}^{i+C} \sum_{x=1}^{ m_A} \mu(x,y^j) \sum_{a, b = 1}^{d} V(a, \sigma_{x,y^j}(a)| x, y^j) P(a, b^j | x, y^j) \nonumber \\
\end{eqnarray}
is the expression with the contradictions in the Bell inequality spread over $C$ other parties each of which measures a single setting (note that the sum over $j$ is taken modulo $C+1$). It is also clear following the argument in the previous proof that the maximum local value of $\mathcal{B}^{U}_{AB}$ is equal to the maximum local value of $\mathcal{\underline{B}}^{U}$, i.e. 
\begin{equation}
\label{local-eq2}
\mathit{R}_L(\mathbb{B}^{U}_{AB}) = \mathit{R}_L(\mathbb{\underline{B}}^{U}).
\end{equation}
%
We now make the following useful observation.
%

\begin{obs}
\label{obs-local-marg}
The maximum value of $\mathcal{\underline{B}}^{U, nc}$ over all no-signaling boxes under the constraint of an arbitrary fixed marginal distribution $\{\tilde{P}^{A}(a|x)\}$ and with $\mu(x,y) = \beta_{y,y'} \mu(x,y')$ is achieved by a local box $\sum_k p_k \{P^{A}_{k}(a|x)\} \otimes \{P^{B}_{k}(b|y)\}$ with $\sum_k p_k P^{A}_k(a|x) = \tilde{P}^{A}(a|x) \; \; \forall a, x$.
\end{obs}

\begin{proof}
By definition, $\mathcal{\underline{B}}^{U, nc}$ without any marginal constraints is saturated by a local deterministic box $\{P_D(a,b|x,y)\} = \{P^{A}_D(a|x)\} \otimes \{P^{B}_D(b|y)\}$ for every outcome $a$ for every setting $x$ of Alice, i.e. for every $a^* \in [d], x^* \in [m_A]$ there exists a deterministic box that outputs $P^{A}_D(a^*|x^*) = 1$ and saturates the Bell expression. 

Consider the $d$ deterministic boxes $P_D^i$ ($i \in [d]$) that output with certainty the outcome $a^1_i \in [d]$ for input $x = 1$ and saturate the expression $\mathcal{\underline{B}}^{U, nc}$. Due to the fact that they win this part of the unique game with certainty, each box $P_D^i$ outputs a different outcome for all the other settings of Alice and settings of Bob, i.e. 
\begin{eqnarray}
&&a^x_i \neq a^{x}_{i'} : \; \; \; i \neq i', x \neq 1, \nonumber \\ 
&&b^y_i \neq b^{y}_{i'} : \; \; \; i \neq i', y \in [m_B - C^{(s)}].
\end{eqnarray}
This implies that there exists a local relabeling of the outputs by Bob and for other settings by Alice such that the boxes $P_D^i$ can therefore be rewritten as deterministically returning $a = b = i$ for all $x \in [m_A], y \in [m_B - C^{(s)}]$, i.e.
\begin{equation}
P_D^i(a=i|x) = 1, \; \; P_D^i(b=i|y) = 1 \; \; \forall x, y. 
\end{equation}
Therefore, winning with certainty the part of the game corresponding to $\mathbb{\underline{B}}^{nc}$ is equivalent up to a local relabeling to winning a game that only involves the identity permutation, i.e. with $\sigma_{x,y}(a) = a \; \; \forall x \in [m_A], y \in [m_B - C]$. This implies that we only need to show the statement for the game involving the identity permutation $\mathbb{\underline{B}}^{nc}_{\openone}$. 

For the identity game $\mathbb{\underline{B}}^{nc}_{\openone}$ with $\mu(x,y) = \beta_{y,y'} \mu(x,y')$ for parameters $\beta_{y,y'}$ that do not depend on $x$, it can be readily seen that the maximization over all no-signaling boxes with fixed marginal $\{\tilde{P}^A(a|x)\}$ is achieved by a box $\{P_L(a,b|x,y)\}$ with $$\{P_L(a, b|x, y= y_1)\} = \{P_L(a,b|x, y = y_2)\} \; \; \forall y_1, y_2.$$ Such a box is manifestly local, i.e. $$\{P_L(a,b|x,y)\} = \sum_k p_k \{P^{A}_{L, k}(a|x)\} \otimes \{P^{B}_{L, k}(b|y)\}$$ with $$\sum_k p_k P^{A}_{L, k}(a|x) = \tilde{P}^{A}(a|x) \; \; \forall a, x.$$ This proves the Observation \ref{obs-local-marg}.
\end{proof}

Now, let $\{\underline{Q}(a, b^1, \dots, b^{C+1} | x, y^1, \dots, y^{C+1})\}$ denote the box that maximizes the expression $\mathbb{\underline{B}}_i$ over all no-signaling boxes. Since all parties except the $i^{th}$ party measure a single setting for this expression, this box takes the form
\begin{eqnarray}
&&\{\underline{Q}(a,b^1, \dots, b^{C+1} | x, y^1, \dots,  y^{C+1})\} = \nonumber \\ && \{\underline{Q}^{1}(b^1|  y^1)\} \otimes \dots  \{\underline{Q}^{i}(a,b^i|x, y^i)\}  \dots \otimes \{\underline{Q}^{C+1}(b^{C+1} | y^{C+1})\}. \nonumber 
\end{eqnarray}
The value of the expression $\mathbb{\underline{B}}^{nc}_i$ depends on $\{\underline{Q}^{i}(a,b^i|x,y^i)\}$ while the value of $\mathbb{\underline{B}}^{c}_i$ depends on the rest of the boxes together with Alice's marginal $\{\underline{Q}^{A}(a|x)\}$. Now, by Observation (\ref{obs-local-marg}), we know that the box $\{\underline{Q}^{i}(a,b^i|x,y^i)\}$ can be replaced by a local box 
\begin{equation}
\{\underline{Q}^{i}_L(a,b^i|x,y^i)\} = \sum_k q_k \underline{Q}^{A}_{L, k}(a|x) \otimes \underline{Q}^{B^i}_{L, k}(b^i|y^i)
\end{equation}
with the same marginal on Alice's side, $\sum_k q_k \underline{Q}^{A}_{L, k}(a|x) = Q^{A}(a|x)$. Therefore, the maximum value of $\mathbb{\underline{B}}_i.\{P(a, b^1, \dots, b^{C+1}|x, y^1, \dots, y^{C+1})\}$ over all no-signaling boxes is achieved at a local box. The above together with Eq.(\ref{local-eq2}) shows that the maximum over all no-signaling boxes of $\sum_i \mathcal{\underline{B}}^U_i = \sum_i \mathcal{B}^U_{AB^{i}}$ cannot be larger than $(C^{(s)}+1) \mathit{R}_L(\mathbb{B}^U_{AB})$. This proves the statement for arbitrary $m_A, m_B, d$ and restricted $\mu(x,y)$.

We now proceed to prove the second statement in the Proposition, concerning unique games with $m_A = 2$ and arbitrary $m_B,  d, \mu(x,y)$ denoted by $\mathbb{B}^{m_A = 2}$ where we have dropped the superscript $U$ for simplicity. As before, $\sum_i \mathbb{B}^{m_A = 2}_{AB^{i}}$ can be rewritten as $\sum_i \mathbb{\underline{B}}^{m_A =2}_i$ with $\mathbb{\underline{B}}^{m_A = 2}_i = \mathbb{\underline{B}}^{m_A = 2, nc}_i + \mathbb{\underline{B}}^{m_A = 2, c}_i$ with the corresponding expressions given as in Eq. (\ref{unique-no-con}) and Eq. (\ref{unique-con}) with $m_A = 2$. The difficulty is in proving the analogous statement to Observation (\ref{obs-local-marg}) which does not hold for arbitrary $m_A$ without the constraint on $\mu(x,y)$ which is why we restrict to the case $m_A = 2$.

\begin{lemma}
\label{obs-local-marg2}
The maximum value of $\mathcal{\underline{B}}^{m_A = 2, nc}$ with
\begin{eqnarray}
\label{unique-no-con2}
&&\mathcal{\underline{B}}^{m_A = 2, nc} \defeq  \nonumber \\ &&\sum_{x=1}^{2} \sum_{y=1}^{m_B - C^{(s)}} \mu(x,y) \sum_{a, b = 1}^{d} V(a, \sigma_{x,y}(a)| x, y) P(a, b | x, y) \nonumber \\
\end{eqnarray}
over all no-signaling boxes under the constraint of a fixed marginal distribution $Q^{A}(a|x)$ is achieved by a local box $\sum_k p_k P^{A}_{k}(a|x) \otimes P^{B}_{k}(b|y)$ with $\sum_k p_k P^{A}_k(a|x) = Q^{A}(a|x)$ for arbitrary $m_B, d, \mu(x,y)$. 
\end{lemma}
\begin{proof}
To prove the statement, we shall show that the maximization of $\mathcal{\underline{B}}^{m_A = 2, nc}$ can be done for a single setting $y$ of Bob, and that the maxima for different $y$ result in compatible local boxes. As seen in the first part of the proof, it is sufficient to consider the case that the unique game without contradictions only involves the identity permutation, i.e., $\sigma_{x,y}(a) = a \; \; \forall x \in [m_A], y \in [m_B - C]$.

Let us consider the maximization over no-signaling boxes of the expression with a single setting for Bob, i.e. $$ \sum_{x=1}^{2} \mu(x,y=1) \sum_{a, b = 1}^{d} V(a, a| x, y=1) P(a, b | x, y=1)$$ for fixed marginals $Q^A(a|x)$. This is naturally maximized by a local box which can be expressed as a convex combination of deterministic boxes. Denote by $p^{b}_{a_1, a_2}$ the coefficient in this convex decomposition of the deterministic box that outputs with certainty $a_1, a_2$ for the two settings of Alice and $b$ for the single setting of Bob. 
Denoting $\mu(1,1) =: w_1$ and $\mu(2,1) =: w_2$, the value of the expression with local boxes is given by
\begin{eqnarray}
\label{identity-game}
&&\sum_{x=1}^{2} \sum_{a, b = 1}^{d} w_x V(a, a|x,y=1) P(a, b | x, y) =  \nonumber \\ &&\sum_{a = 1}^{d} p^{a}_{a,a} (w_1 + w_2) + \sum_{a_1 = 1}^{d} \sum_{\substack{a_2 = 1 \\ a_2 \neq a_1}}^{d} \left(p^{a_1}_{a_1,a_2} w_1 +  p^{a_1}_{a_2,a_1} w_2\right), \nonumber
\end{eqnarray}
where we have used the fact that $p^{b}_{a_1, a_2} = 0$ for $b \neq a_1, b\neq a_2$ in this maximization.
Since the box satisfies the marginal constraints, the following equalities hold
\begin{eqnarray}
&&Q^A(a|x=1) = \sum_{b, a_2} p^{b}_{a,a_2}, \nonumber \\
&&Q^A(a|x=2) = \sum_{b, a_1} p^{b}_{a_1,a}.
\end{eqnarray}
For the identity game, it is readily seen that to achieve the maximum, we should have
\begin{equation}
p^{a}_{a,a} = \min_{x} Q^A(a|x)
\end{equation}
due to the fact that this term appears with the highest weight $w_1 + w_2$ in Eq.(\ref{identity-game}).
The marginal constraints can thus be written explicitly as
\begin{eqnarray}
\label{marginal-cons}
&& Q^A(a|1) - \min_{x} Q^A(a|x) = \sum_{\substack{a_2=1 \\ a_2 \neq a}}^{d} \left( p^{a}_{a, a_2} + p^{a_2}_{a, a_2} \right), \nonumber \\
&& Q^A(a|2) - \min_{x} Q^A(a|x) =\sum_{\substack{a_1=1 \\ a_1 \neq a}}^{d} \left( p^{a}_{a_1, a} + p^{a_1}_{a_1,a} \right).
\end{eqnarray}
Denoting $s_{a_1, a_2} \defeq p^{a_1}_{a_1, a_2} + p^{a_2}_{a_1, a_2}$, we see that the sum 
\begin{equation}
\sum_{\substack{a_1, a_2 = 1 \\ a_1 \neq a_2}}^{d} s_{a_1, a_2} = 1 - \sum_{a=1}^{d} \min_{x} Q^A(a|x)
\end{equation}
is fixed for a given marginal box $\{Q^A(a|x)\}$.
Therefore, we can rewrite the value of the Bell expression with single setting for Bob by fixing the set of coefficients $\{s_{a_1, a_2}\}$ to satisfy Eq.(\ref{marginal-cons}) and maximizing over a set of parameters $\{g_{a_1, a_2}\}$ with $0 \leq g_{a_1, a_2} \leq 1$ as 
\begin{eqnarray}
&&\sum_{x=1}^{2} \sum_{a, b = 1}^{d} w_x V(a, a| x, y=1) P(a, b | x, y=1) = \nonumber \\ &&\sum_{a = 1}^{d} \min_{x} Q^A(a|x) (w_1 + w_2) + \sum_{\substack{a_1, a_2=1 \\a_1 \neq a_2}}^{d} s_{a_1,a_2} w_2 \nonumber \\ && \qquad \qquad \qquad \qquad +\max_{\{g_{a_1, a_2}\}} \sum_{\substack{a_1, a_2=1 \\a_1 \neq a_2}}^{d} g_{a_1, a_2} s_{a_1, a_2} (w_1 - w_2). \nonumber 
\end{eqnarray}
Note that the same value of $$\sum_{\substack{a_1, a_2 = 1 \\ a_1 \neq a_2}}^{d} \Big( g_{a_1, a_2} s_{a_1, a_2} w_1 + (1 - g_{a_1, a_2}) s_{a_1, a_2} w_2 \Big)$$ for fixed marginals can be obtained for any set $\{s_{a_1, a_2}\}$ satisfying the marginal constraints Eq.(\ref{marginal-cons}) by appropriately choosing $g_{a_1, a_2}$. In other words, any value obtained for a set $\{s_{a_1, a_2}\}$ satisfying Eq.(\ref{marginal-cons}) and a set of coefficients $\{g_{a_1, a_2}\}$ can equally be obtained for any other set $\{s^*_{a_1, a_2}\}$ satisfying Eq.(\ref{marginal-cons}) by choosing $\{g^*_{a_1, a_2}\}$ to solve the system of equations
\begin{equation}
\sum_{\substack{a_1, a_2 = 1 \\ a_1 \neq a_2}}^{d} g^*_{a_1, a_2} s^*_{a_1, a_2}  = \sum_{\substack{a_1, a_2 = 1 \\ a_1 \neq a_2}}^{d} g_{a_1, a_2} s_{a_1, a_2} \; \; \forall \{a_1, a_2\}
\end{equation}
This system always has a solution since 
\begin{equation}
\sum_{\substack{a_1, a_2 = 1 \\ a_1 \neq a_2}}^{d} s^*_{a_1, a_2} = \sum_{\substack{a_1, a_2 = 1 \\ a_1 \neq a_2}}^{d} s_{a_1, a_2}
\end{equation}
and $0 \leq g_{a_1, a_2} \leq 1$ and $0 \leq g^*_{a_1, a_2} \leq 1$. The above arguments can be repeated for every setting $y$ of Bob and we see that in the expression $$\sum_{x=1}^{2} \sum_{y=1}^{m_B - C^{(s)}} \mu(x,y) \sum_{a, b = 1}^{d} V(a, a| x, y) P(a, b | x, y)$$ the maximum value for each setting of Bob is achieved by a local box whose convex decomposition can be written in terms of a fixed set $\{s_{a_1, a_2}\}$. It remains to show that this implies that the maximum of the whole expression is achieved by a local box.  

Let us use the label $p^{b}_{a_1, a_2}(y)$ for the coefficients of the deterministic boxes that appear in the maximization procedure of the Bell expression restricted to setting $y$ as above. We now construct a local box which gives the maximum value for the whole expression and use the notation $p^{b_1, \dots, b^{m_B - C^{(s)}}}_{a_1, a_2}$ for the coefficients of the deterministic boxes in its convex combination. The above arguments show that $p^{a}_{a,a}(y) = \min_{x} Q^{A}(a|x) \; \; \forall y$, so we may set 
$p^{a, \dots, a}_{a,a} = \min_{x}Q^{A}(a|x)$. For $a_1 \neq a_2$, we proceed by induction. For given $s_{a_1, a_2}$ assume that we have constructed a set of deterministic boxes with coefficients $p^{b_1, \dots, b_k}_{a_1, a_2}$ ($b^i \in \{a_1, a_2\}$ arranged in lexicographic order) for fixed $k$ such that for each $y \in [k]$, the boxes maximize the Bell expression with setting $y$ alone. Note that this implies that $$\sum_{b_1, \dots, b_k} p^{b_1, \dots, b_k}_{a_1, a_2} = s_{a_1, a_2}.$$ For any $p^{a_1}_{a_1, a_2}(k+1)$ and $p^{a_2}_{a_1, a_2}(k+1)$ satisfying $ p^{a_1}_{a_1, a_2}(k+1) + p^{a_2}_{a_1, a_2}(k+1) = s_{a_1, a_2}$ we now construct suitable $p^{b_1, \dots, b_{k+1}}_{a_1, a_2}$ such that the Bell expression for setting $k+1$ is also maximized in addition to that for the previous $k$ settings. To do this, we find suitable $0 \leq \gamma \leq 1$ and $\tilde{p}^{b_1, \dots, b_k}_{a_1, a_2}$ such that for a set $W$ of coefficients lower than ${\tilde{b}_1, \dots, \tilde{b}_k}_{a_1, a_2}$ in the lexicographic order ($W \defeq \{\{b_1, \dots, b_k\} : \{b_1, \dots, b_k\} <_L \{\tilde{b}_1, \dots, \tilde{b}_k\}\}$) we have
\begin{eqnarray}
&&\sum_{\{b_1, \dots, b_k\} \in W} p^{b_1, \dots, b_k}_{a_1, a_2} + \gamma {p}^{\tilde{b}_1, \dots, \tilde{b}_k}_{a_1, a_2} = p^{a_1}_{a_1, a_2}(k+1) \nonumber \\
&&\sum_{\substack{\{b_1, \dots, b_k\} \in \overline{W} \\ \{b_1, \dots, b_k\} \neq \{\tilde{b}_1, \dots, \tilde{b}_k\}}} p^{b_1, \dots, b_k}_{a_1, a_2} + (1 - \gamma) p^{\tilde{b}_1, \dots, \tilde{b}_k}_{a_1, a_2} = p^{a_2}_{a_1, a_2}(k+1). \nonumber
\end{eqnarray}
We can now finish the construction of the local box by assigning the coefficients $p^{b_1, \dots, b_{k+1}}_{a_1, a_2}$ as
\begin{widetext}
\begin{displaymath}
   p^{b_1, \dots, b_k, b_{k+1} }_{a_1, a_2} = \left\{
     \begin{array}{lr}
       p^{b_1, \dots, b_k}_{a_1, a_2} & : (b_{k+1} = a_1) \wedge \{b_1, \dots, b_k\} \in W\\
       \gamma p^{\tilde{b}_1, \dots, \tilde{b}_k}_{a_1, a_2}& : (b_{k+1} = a_1) \wedge b_i = \tilde{b}_i \; \; \forall i \in [k] \\
       (1 - \gamma) p^{\tilde{b}_1, \dots, \tilde{b}_k}_{a_1, a_2}& : (b_{k+1} = a_2) \wedge b_i = \tilde{b}_i \; \; \forall i \in [k] \\
       p^{b_1, \dots, b_k}_{a_1, a_2} & : (b_{k+1} = a_2) \wedge \{b_1, \dots, b_k\} \in \overline{W} \setminus \{ \tilde{b}_1, \dots, \tilde{b}_k\} \\
     \end{array}
   \right.
\end{displaymath} 
\end{widetext}
Following this procedure inductively for all $y$, we construct a local box that maximizes the Bell value for each setting $y \in [m_B - C^{(s)}]$ of Bob. Therefore, a local box maximizes the identity game expression in Eq. (\ref{unique-no-con2}) with the given marginal distribution $\{Q^A(a|x)\}$. 

\end{proof}
Now following the previous proof, we see that any box 
\begin{eqnarray}
&&\{\underline{Q}(a,b^1, \dots, b^{C+1} | x, y^1, \dots,  y^{C+1})\} = \nonumber \\ &&\{\underline{Q}^1(b^1| y^1)\} \otimes \dots  \{\underline{Q}^i(a,b^i|x, y^i)\}  \dots \otimes \{\underline{Q}^{C+1}(b^{C+1} | y^{C+1})\} \nonumber 
\end{eqnarray}
that achieves the maximal no-signaling value of the unique game expression $\mathcal{\underline{B}}^{m_A =2}_i$
with arbitrary $d, m_B, \mu(x,y)$ can be replaced by a local box that achieves the same value. This is done by replacing the part of the box $\{\underline{Q}^i(a,b^i|x, y^i)\}$ with a local box from the construction in Lemma \ref{obs-local-marg2} which achieves the same value for $\mathcal{\underline{B}}^{m_A = 2, nc}$ while at the same time attaining the same marginal $\{Q^A(a|x)\}$ so that the value of $\mathcal{\underline{B}}^{m_A = 2, c}$ is also left unchanged. Therefore, the maximum value of $\sum_i \mathbb{B}^{U, m_A = 2}_{AB^{i}}$ over all no-signaling boxes cannot be larger than $(\mathit{C}^{(s)} + 1)\mathit{R}_L(\mathbb{B}^{U, m_A = 2})$. Note that the statement in Lemma (\ref{obs-local-marg2}) does not hold for arbitrary marginal distributions for larger number of settings for Alice so the analogous monogamy relation for $m_A > 2$ may not hold in general. This ends the proof of Proposition 2.
\end{proof}

\textbf{Observation 1.}
\textit{For any unique game $\mathbb{B}^U$ with parameters $m_A, m_B, d$ and 
\begin{equation}
\label{nosig-adv}
\mathit{R}_{NS}(\mathbb{B}^U) > \mathit{R}_L(\mathbb{B}^U)
\end{equation}
and any no-signaling box $\{P(a,b,c|x,y,z)\}$ with $a,b,c \in [d]$ and $x \in [m_A], y,z \in [m_B]$, we have
\begin{equation}
\label{weak-mono-unique}
(\mathbb{B}_{AB}^U + \mathbb{B}_{AC}^U).\{P(a,b,c|x,y,z)\} < 2 \mathit{R}_{NS}(\mathbb{B}^U).
\end{equation}}
\begin{proof}
Consider the general unique game inequality
\begin{eqnarray}
&&\mathcal{B}_{AB}^{U} \defeq \nonumber \\ &&\sum_{x=1}^{m_A} \sum_{y=1}^{m_B} \mu(x,y) \sum_{a, b = 1}^{d} V(a, \sigma_{x,y}(a)| x, y) P(a, b | x, y) \leq \mathit{R}_L \nonumber \\
\end{eqnarray}
with $\mathit{R}_{NS}(\mathbb{B}^U) > \mathit{R}_L(\mathbb{B}^U)$. Being a non-trivial game $\mathbb{B}^U$ has at least one setting $y^*$ of Bob which leads to a contradiction with local realistic predictions. We may rewrite $(\mathcal{B}_{AB}^U + \mathcal{B}_{AC}^U)$ as $\sum_{i=1,2} \mathcal{\underline{B}}^U_i$ with $\mathcal{\underline{B}}^U_i = \mathcal{\underline{B}}^{U, m_B-1}_i + \mathcal{\underline{B}}^{U, y^*}_i$.
Here
\begin{equation}
\mathcal{\underline{B}}^{U, m_B-1}_1 = \sum_{x=1}^{m_A} \sum_{\substack{y =1 \\ y \neq y^*}}^{m_B} \mu(x,y) \sum_{a, b = 1}^{d} V(a, \sigma_{x,y}(a)| x, y) P(a, b | x, y)
\end{equation}
is an expression containing the $m_B - 1$ settings other than $y^*$ in the Bell expression $\mathbb{B}_{AB}^U$ evaluated on the Alice-Bob marginal $\{P(a,b|x,y)\}$ and
\begin{equation}
\mathcal{\underline{B}}^{U, y^*}_1 = \sum_{x=1}^{m_A} \mu(x,y^*) \sum_{a, c = 1}^{d} V(a, \sigma_{x,y^*}(a)| x, y^*) P(a, c | x, y^*).
\end{equation}
is an expression with the remaining setting $y^*$ evaluated on the Alice-Charlie marginal $\{P(a,c|x,z)\}$.
$\mathcal{\underline{B}}^{U, m_B-1}_2$ and $\mathcal{\underline{B}}^{U, y^*}_2$ are defined similarly by interchanging the roles of Bob and Charlie. 

We have seen that the maximum algebraic value $\mathit{R}_{NS}(\mathcal{B}_{AB}^U)$ with some fixed $\{\sigma_{x,y}\}$ can be achieved in no-signaling theories by the box with non-zero entries $P(a, \sigma_{x,y}(a) | x,y) = \frac{1}{d} \; \; \forall a \in [d]$. We now show that this value cannot be reached for $\mathcal{\underline{B}}^U_1$ and consequently that $\sum_{i=1,2} \mathcal{\underline{B}}^U_i$ cannot equal $2 \mathit{R}_{NS}(\mathcal{B}_{AB}^U)$ in any no-signaling theory. Without loss of generality, the box maximizing $\mathcal{\underline{B}}^U_1$ can be written as
\begin{equation}
\{Q(a,b,c|x,y,z=y^*)\} = \{Q^{AB}(a,b|x,y)\} \otimes \{Q^{C}(c|z=y^*)\}
\end{equation}
where $\{Q^{C}(c|z=y^*)\}$ is a deterministic box that returns a fixed output $c^*$ upon input $z = y^*$. Since the Bell expression involves a fixed outcome of Alice $a_x^* = \sigma_{x,y^*}^{-1}(c^*)$ corresponding to $c^*$ for every input $x$, the algebraic value can only be attained if these occur deterministically, i.e. for algebraic violation the box $\{Q^{AB}(a,b|x,y)\}$ should satisfy
\begin{equation}
Q^A(a_x^*|x) = 1 \; \; \; \forall x.
\end{equation}
This in turn implies that for the other settings $y \neq y^*$, in order to have algebraic violation the output $\sigma_{x,y}(a_x^*)$ must be deterministic, or in other words that the box $\{Q^{AB}(a,b|x,y)\}$ is deterministic. But this would imply $\mathit{R}_{L}(\mathbb{\underline{B}}^U_1) = \mathit{R}_{NS}(\mathbb{B}^U_{AB})$ which leads to a contradiction with Eq.(\ref{nosig-adv}) since $\mathit{R}_{L}(\mathbb{\underline{B}}^U_1) = \mathit{R}_{L}(\mathbb{B}^U_{AB})$ (recall that any deterministic strategy for $\underline{B}^U_1$ can be recast as a deterministic strategy for $\mathbb{B}^U_{AB}$ and vice-versa). This shows that 
\begin{equation}
\sum_{i=1,2} \mathcal{\underline{B}}^U_i < 2 \mathit{R}_{NS}(\mathcal{B}_{AB}^U)
\end{equation}
which is equivalent to the assertion in Eq.(\ref{weak-mono-unique}).
\end{proof}

\textbf{Claim 3.}
\textit{There exists a two-party Bell expression $\mathbb{B}$ with $m_A = m_B = 3, d = 4$, associated contradiction numbers $C_{\mathbb{B}} = C^{(s)}_{\mathbb{B}} = 1$, maximum no-signaling value greater than the maximum local value ($\mathit{R}_{NS} > \mathit{R}_{L}$) and a three-party no-signaling box $\{P(a,b,c|x,y,z)\}$ with $a, b, c \in \{1, 2, 3, 4\}$ and $x, y, z \in \{1, 2, 3\}$ such that 
\begin{equation}
(\mathbb{B}_{AB} + \mathbb{B}_{AC}).\{P(a,b,c|x,y,z)\} = 2 \mathit{R}_{NS}.
\end{equation}}

\begin{proof}
The two-party Bell inequality is from \cite{Cabello2} and involves each party measuring one of three settings and obtaining one of four outcomes. The indicator vector $\mathbb{B}$ representing the inequality can be represented in the form in Table \ref{table:counterexample}. 
\begin{table}[t]

\begin{center}
  \begin{tabular}{| l | l | l | c | r | }
    \hline
  &    B &       I       &     II      &      III      \\ \hline
 A  &    & 1 2 3 4  & 1 2 3 4 & 1 2 3 4 \\ \hline
    &1 & + + - - & + + - - & + + - - \\ 
    I&2 & + + - - & -  - + + & - - + + \\ 
  & 3 & - - + + & + + - - & - - + + \\ 
    & 4 & - - + + & - - + + & + + - - \\ \hline
   &1 & + - + - & + - + - & + - + - \\ 
   II&2 & + - + - & - + - + & - + - + \\
   &3 & - + - + & + - + - & - + - + \\
   &4 & - + - + & - + - + & + - + - \\ \hline
   &1 & + - - + & + - - + & - + + - \\
  III&2 & + - - + & - + + - & + - - + \\
   &3 & - + + - & + - - + & + - - + \\
   &4 & - + + - & - + + - & - + + - \\ \hline
  \end{tabular}
\end{center}
\caption{\label{table:counterexample} Table representing the Bell inequality that presents a counterexample to the strong monogamy relation}
\end{table}

Here, the settings of Alice (party $A$) labeled $I, II$ and $III$ are represented vertically with the rows corresponding to the outcomes $1, 2, 3$ and $4$ for each setting. Similarly, the settings of Bob (party $B$) labeled $I, II$ and $III$ are represented horizontally with the columns corresponding to the outcomes $1, 2, 3$ and $4$ for each setting. The entries $+$ and $-$ denote the coefficients appearing in front of the corresponding probabilities in the Bell expression with $\mathbb{B}(a,b,x,y) = 1$ for entry $+$ and $\mathbb{B}(a,b,x,y) = 0$ for the $-$ entry. One can check that the maximum classical value of $\mathbb{B}.\{P(a,b|x,y)\}$ is $8$ while the maximum no-signaling value (which incidentally  \cite{Cabello2} is also the maximum quantum value) is $9$. One can also check that if Bob does not measure setting $III$, so that $x \in \{I, II, III\}$ and $y \in \{I, II\}$, there is a classical box with entries $P(1,1|x,y) = 1$ for all $x, y$
which achieves the maximum no-signaling value of $6$ for the remaining Bell expression. Therefore, the contradiction number for this Bell inequality is one, i.e. $C_{\mathbb{B}} = 1$. Moreover, it is also straightforward to check that for every outcome $a$ and for every setting $x$ of Alice, there exists a deterministic box with $P^A(a|x) = 1$ which achieves the maximum no-signaling value of $6$ for the expression, so that $C^{(s)}_{\mathbb{B}} = 1$ as well.

Crucially however, there is a tripartite no-signaling box $\{P(a,b,c|x,y,z)\}$ with $a, b, c \in \{1, 2, 3, 4\}$ and $x, y, z \in \{I, II, III\}$, with marginals $\{P(a,b|x,y)\}$ and $\{P(a,c|x,z)\}$ each of which leads to the maximum violation of the inequality, i.e. $(\mathbb{B}_{AB} + \mathbb{B}_{AC}).\{P(a,b,c|x,y,z)\} = 18$. The entries of the box $\{P(a,b,c|x,y,z)\}$ are given by:
\begin{widetext}
\begin{displaymath}
   P(a,b,c|x,y,z) = \left\{
     \begin{array}{lr}
       \frac{1}{8} & : (y=z) \wedge (b = c) \wedge  \mathbb{B}(a,b,x,y) = +1\\
       \frac{1}{16} & : (y \neq z) \wedge \mathbb{B}(a,b,x,y) = +1 \wedge \mathbb{B}(a,c,x,z) = +1 \\
       0 & : otherwise
     \end{array}
   \right.
\end{displaymath} 
\end{widetext}
One can check that the box explicitly obeys all the no-signaling constraints in Eq. (\ref{no-signal}). 
Moreover, it can be seen that the marginals $\{P(a,b|x,y)\}$ and $\{P(a,c|x,z)\}$ obey
\begin{displaymath}
   P(a,b(c)|x,y(z)) = \left\{
     \begin{array}{lr}
       \frac{1}{8} & : \mathbb{B}(a,b (c),x,y (z)) = +1\\
        0 & : \mathbb{B}(a, b(c), x, y (z)) = 0. 
     \end{array}
   \right.
\end{displaymath} 
which leads to $\mathbb{B}.\{P(a,b|x,y)\} = \mathbb{B}.\{P(a,c|x,z)\} = 9$.
In other words, both Alice-Bob and Alice-Charlie are able to violate the inequality to its algebraic maximum value within the no-signaling theory. 
\end{proof}

\end{document}